\documentclass[copyright,creativecommons]{eptcs}
\providecommand{\event}{FMSPLE 2016}

\usepackage{monotonicity-style}
\usepackage{breakurl}
\title{Refactoring Delta-Oriented Product Lines to Achieve Monotonicity%
  \thanks{The authors of this paper are listed in alphabetical order. 
This work has been partially supported by
  project HyVar (\url{www.hyvar-project.eu}), which has received funding from the European Union's Horizon 2020 research and
 innovation programme under grant agreement No.\ 644298; by ICT COST Action IC1402 ARVI (\url{www.cost-arvi.eu}); and by Ateneo/CSP D16D15000360005 project RunVar.}}
\def\titlerunning{Refactoring Delta-Oriented Product Lines to Achieve Monotonicity}
\def\authorrunning{F. Damiani \& M. Lienhardt}

\author{Ferruccio Damiani
\institute{University of Turin (Italy)}
\email{ferruccio.damiani@di.unito.it}
\and
Michael Lienhardt
\institute{University of Turin (Italy)}
\email{michael.lienhardt@di.unito.it}
}
\date{February 8, 2016}
\begin{document}

\maketitle

% I need to change a little bit the abstract to take in account the fact that best practices have already been studied a lot

\begin{abstract}
Delta-oriented programming (DOP) is a flexible transformational approach to implement software product lines.
In delta-oriented product lines, variants are generated by applying operations contained in delta modules to a (possibly empty) base program.
These operations can add, remove or modify named elements in a program (e.g., classes, methods and fields in a Java program).
This paper presents algorithms for refactoring a delta-oriented product line into monotonic form,
 i.e., either to contain add and modify operations only (monotonic increasing) or to contain remove and modify operations only (monotonic decreasing).
Because of their simpler structure, monotonic delta-oriented product lines are easier to %understand and 
analyze.
The algorithms are formalized by means of a core calculus for DOP of product lines of Java programs and their correctness and complexity are given.
%The algorithms are formalized by means of a core calculus for DOP of product lines of Java programs and their correctness and complexity are proven.
%
%This paper presents a refactoring algorithm for Delta-Oriented Programming (\DOP) transforming a Software Product Line (\SPL) into a {\em monotonic} equivalent.
%A monotonic \SPL\ includes a limited set of operations (either $\kwadds$ and $\kwmodifies$, or $\kwremoves$ and $\kwmodifies$),
%  and usually is easier to read, understand and analyse, and is more efficient during variant generation compared to \SPL\ containing all \DOP\ operations.
%We prove the correctness and the complexity of this algorithm and
% illustrate it on the Expression \SPL\ benchmark, showing that the resulting \SPL\ is very close to the original one.
\end{abstract}

\section{Introduction}

%{\em Software Product Lines} (\SPL) are a common way to generate similar programs, called {\em variants},
% from a common code base and a well documented variability~\cite{Clements:2001}. 
%\emph{Delta-Oriented Programming} (\DOP)~\cite{SPLC,BetDamSch:ACTA-2013} is a  flexible transformational and modular approach to implement \SPL s.
%\DOP\ product lines are build around a (possibly empty) base program, to which are added a set of {\em delta modules} implementing the different aspects of the \SPL's variability.
%Therefore (see, e.g.,~\cite{FOSD}), \DOP\ supports both proactive \SPL\ development (i.e., planning all products/variants in advance)
% and extractive \SPL\ development~\cite{Krueger02} (i.e., starting from existing programs).
%
%One of the key feature of \DOP\ is the capability of delta modules to {\em add}, {\em modify} and {\em remove} content to and from the base program.
%Such expressivity allows for very flexible \SPL\ construction,
% e.g. having a base program implementing a large set of features (called {\em complex core} in~\cite{SPLC})
% with delta modules removing parts of it when these features are not selected.
%However, such flexbility can lead to a multiplication of opposite add and remove operations,
% making the \SPL\ unnecessarily complex, difficult to understand and analyze. 
%A classic cause of such complexity is {\em software evolution} and {\em aging}, where modifying the features of the \SPL\ and their implementation
%  can simply be done by adding new patch-like delta modules on top of existing ones.

A {\em Software Product Line} (\SPL) is a set of similar programs, called {\em variants}, that have a well documented variability and are generated from
 a common code base~\cite{Clements:2001}. 
\emph{Delta-Oriented Programming} (\DOP)~\cite{SPLC,BetDamSch:ACTA-2013} is a  flexible and modular transformational approach to implement \SPL s.
A \DOP\ product line comprises a \emph{Feature Model} (\FM), a \emph{Configuration Knowledge} (\CK), and an \emph{Artifact Base} (\AB).
The \FM\ provides an abstract description of variants in terms of \emph{features} (each representing an abstract description of functionality):
 each variant is described by a set of features, called a \emph{product}. 
The \AB\ provides the (language dependent) code artifacts used to build the variants, namely:
 a (possibly empty) base program from which variants are obtained by applying program transformations, described by {\em delta modules}, that can add, remove or modify code.
The \CK\ provides  a mapping from products to variants by describing the connection between the code artifacts in the \AB\ and the features in the \FM:
 it associates to each delta module an \emph{activation condition} over the features and specifies an \emph{application ordering} between delta modules.
\DOP\ supports automated product derivation, i.e., once the features of a product are selected,
 the corresponding variant is generated by applying the activated delta modules to the base program according to the application ordering.

Delta modules are constructed from {\em delta operations} that can {\em add}, {\em modify} and {\em remove} content to and from the base program
 (e.g., for Java programs, a delta module can add, remove or modify classes interfaces, fields and methods). 
As pointed out in~\cite{FOSD}, such flexibility allows \DOP\ to support
 {\em proactive} (i.e., planning all products in advance),
 {\em reactive} (i.e., developing an initial \SPL\ comprising a limited set of products and evolving it as soon as new products are needed or new requirements arise),
 and {\em extractive} (i.e., gradually tranforming a set of existing programs into an \SPL) \SPL\ development~\cite{Krueger02}.
%In proactive SPL development, all products are planned in advance.
%This approach fosters a clear overall structure of the SPL implementation, however it requires a high upfront investment and is not always possible
% (since, e.g., not all products might be know when starting the development). 
%In reactive SPL development, development starts with an initial product line that
% (similarly to the \emph{extreme programming} approach to single application development) is evolved in order to deal with changing requirements.
%In extractive SPL development  one or more of existing programs are used as  a baseline from which the other variants are derived.
%
%This flexibility 
\DOP\  allows for quick \SPL\ evolution and extension, as modifying or adding products/variants can straightforwardly be achieved by
 adding to the \SPL\ new %patch-like 
delta modules that modify, remove and add code on top of the original implementation of the \SPL.
However, a number of such \SPL\ evolution and extension phases lead, almost ineluctably, to a multiplication of opposite add and remove operations,
 making the resulting \SPL\ complex, difficult to understand and to analyze~\cite{Schulze:2013:RDS:2451436.2451446}.
% How to add a reference to ina's paper?

%Evolving a product line by exploiting such flexibility, which allows to add new products or to quickly modify the implementation of products 
%by introducing new patch-like delta modules on top of existing ones, may lead to
% unnecessarily complex \SPL\ overall structures (comprising, e.g., several opposite add and remove operations)  that are difficult to understand and analyse.

Refactoring~\cite{DBLP:conf/xpu/Fowler02} is an established technique to reduce complexity and to prevent the process of software aging,
 and consists of program transformations that change the internal structure of a program without altering its external (visible) behavior.
Refactoring for \DOP\ product lines, i.e. changing the internal structure of an \SPL\ without changing its products/variants,
 has been investigated in~\cite{Schulze:2013:RDS:2451436.2451446}.
There, a catalogue of refactoring algorithms and code smells is presented.
Most of these refactorings are based on object-oriented refactorings~\cite{DBLP:conf/xpu/Fowler02}.
In particular, the refactorings that refer to delta modules focus on a single delta module or a pair of delta modules.

%None of them addresses the issue of  eliminaning opposite  operations across the whole product line.
%None of them solve the problem of patch-like delta modules.
%A catalogue of many refactoring agorithms and code smells for \DOP\ were proposed in~\cite{Schulze:2013:RDS:2451436.2451446},
%  but none of them solve the problem of patch-like delta modules.

In this paper, we propose two new refactoring algorithms to automatically  eliminate opposite add and remove operations across the whole \SPL,
 consequently reducing %lowering the number of patch-like delta modules and 
the overall complexity of the refactored \SPL\ and making it easier to analyze.
These algorithms are constructed around the notion of {\em monotonicity} where
 {\em increasing monotonic} \SPL\ corresponds to only adding new content to the base program,
 while {\em decreasing monotonic} \SPL\ corresponds to only removing content from the base program.
These two notions of monotonicity are discussed in Section~\ref{sec:algo}, where we propose several definitions with different versions of these concepts.
%However, as discussed in this paper, these definition are restrictive and more flexible notions must be considered.
%Another main characteristic of these algorithms is that they only modify the input \SPL\ in the two following ways:
% removing and adding delta modules;
% and removing and adding operations inside them.
%This approach has two advantages:
% our algorithms do not duplicate any part of the input \SPL;
% and our algorithms have a low complexity in space and time.
%
%Another main characteristic of these algorithms is that they only modify the input \SPL\ by removing and transforming its delta operations.
%Notably, t
%
The refactoring algorithms do not introduce code duplication in the \AB\ of the refactored SPL and have  at most quadratic
%a low (at most quadratic)
 complexity in space and time.
We formalize the notions of monotonicity and the refactoring algorithms by means of {\sc Imperative Featherweight Delta Java} (\IFDJ)~\cite{BetDamSch:ACTA-2013},
 a core calculus for \DOP\ product lines where variants are written in an imperative  version of \textsc{Featherweight Java} (\FJ)~\cite{IgarashiPierceWadler:TOPLAS-2001}.
%\footnote{\FJ\ is not suitable as the underlying formal model since, for each class, it has a single (implicitly defined) constructor with one parameter for each field of the class.
%  As pointed out in~\cite{Delaware:FSE-2009} (and~\cite{BetDamSch:ACTA-2013}), this badly  mixes with FOP (and DOP)  as fields  can be added (and removed);
%   thus invalidating  all existing constructor calls.} 
%
A prototypical implementation of the refactoring algorithms is available at~\cite{delta-ra-impl-IFDJ}.

 Section~\ref{sec:example} introduces our running example.
% a running example, called \EPL, illustrating the \DOP\ approach to \SPL, and our refactoring algorithms later on;
 Section~\ref{sec:model} recalls \IFDJ.
 Section~\ref{sec:auxdef} introduces some auxiliary notations.
 Section~\ref{sec:algo} illustrates the notions of monotonicity, the refactoring algorithms, and their properties.
%   Section~\ref{sec:conclusion} concludes the paper with a brief discussion of related and future work.
 Section~\ref{sec:related} briefly discusses the related work and Section~\ref{sec:conclusion} concludes the paper.

%
%  Section~\ref{sec:algo} discusses the notions of monotonicity, presents the refactoring algorithms, states their properties and applies them on the \EPL\ example;
%  finally, Section~\ref{sec:conclusion} concludes the paper with a brief discussion of related and future work.

\section{Example}\label{sec:example}

In order to illustrate the monotonicity concept and our refactoring algorithms,
 we use a variant of the \emph{expression product line} (\EPL) benchmark (see, e.g.,~\cite{Lopez-HerrejonBC05,BetDamSch:ACTA-2013}).
We consider the following grammar:
{\small\begin{lstlisting}[basicstyle=\sffamily\footnotesize]
$\;$ Exp  ::=  Lit | Add | Neg   $\qquad$    Lit  ::=  <integers>    $\qquad$    Add ::=  Exp "+" Exp    $\qquad$    Neg ::=  "-" Exp  
\end{lstlisting}\vspace*{-.3em}}
\noindent
Two different operations can be performed on the expressions described by this grammar:
printing, which returns the expression as a string, and evaluating, which returns the value of the expression, either as an int or as a literal expression.
%As described in the introduction, we structure this \EPL\ example in three parts: the feature model, the artifact base and the configuration knowledge.
%As commonly done in \SPL, we structure this \EPL\ example in three parts:
% a \emph{Feature Model} (\FM), a  \emph{Configuration Knowledge} (\CK), and an \emph{Artifact Base} (\AB).
%The \FM\ provides an abstract description of variants in terms of \emph{features}: %~\cite{kang1990}:
%  each feature represents an abstract description of functionality and each variant is identified by a set of features, called a \emph{product}. 
%The \AB\ provides language dependent code artifacts that are used to build the variants:
%  it consists of a \emph{base program} (that might be empty or incomplete) and of a set of \emph{delta modules} which are containers of modifications to a program
%  (e.g., for Java programs, a delta module can add, remove or modify classes and interfaces). 
%The \CK\ connects the code artifacts in the \AB\ with the features in the \FM\ (thus defining a mapping from products to variants):
% it associates to each delta module an \emph{activation condition} over the features and specifies an \emph{application ordering} between delta modules~\cite{SPLC}.

\subsection{The Feature Model}

The functionalities in the EPL can be described by two sets of features:
 the ones concerned with the data are \textsf{Lit} (for literals), \textsf{Add} (for the addition) and \textsf{Neg} (for the negation);
 the ones concerned with the operations are \textsf{Print} (for the classic {\tt toString} method), \textsf{Eval1} (for the {\tt eval} method returning an int)
 and \textsf{Eval2} (for the {\tt eval} method returing a literal expression).
The features \textsf{Lit} and \textsf{Print} are mandatory, while \textsf{Add}, \textsf{Neg}, \textsf{Eval1} and \textsf{Eval2} are optional.
Moreover, as \textsf{Eval1} and \textsf{Eval2} define the same method, they are mutually exclusive.
Figure~\ref{fig:featuremodel} shows the feature model %~\cite{kang1990} 
 of the \EPL\ represented as a feature diagram.
\begin{figure}[tbp]
\centering
\includegraphics[scale=.3]{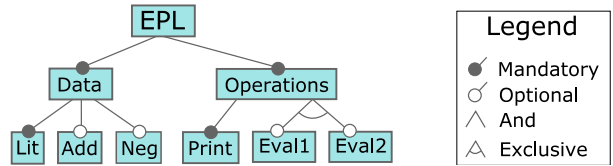}
\caption{Expression Product Line: Feature Model}\label{fig:featuremodel}
\end{figure}

\iffalse
\begin{figure}[t]
\vspace*{-1em}
{\scriptsize$$
\begin{array}{l@{\;\;\;=\;\;\;}l}
\EPL.\gfeatures & \{ \textsf{EPL}, \textsf{Data}, \textsf{Operations}, \textsf{Lit}, \textsf{Add}, \textsf{Neg}, \textsf{Print}, \textsf{Eval1}, \textsf{Eval2} \}
\\
\EPL.\gformula & \textsf{Lit} \land \textsf{Print} \land  \neg (\textsf{fEval1} \land \textsf{fEval2})
\end{array}
$$}
\vspace*{-1.2em}
\caption{Expression Product Line: Feature Model}\label{fig:EPL:FM}
\end{figure}
\fi

\subsection{The Artifact Base}\label{sec:example-AB}

\paragraph{Base Program.}

In our example, the \EPL\ is constructed from the base program shown in Figure~\ref{fig:DLitAddNegPrint}, which is the variant implementing features 
 \textsf{Lit}, \textsf{Add} and \textsf{Print}.
\begin{figure}
\begin{minipage}[t]{.46\textwidth}
\begin{lstlisting}
class Exp extends Object {  // only used as a type
  String toString() { return null; }
}
class Lit extends Exp {
  int value;
  Lit setLit(int n) { value = n; return this; }
  String toString() { return value + ""; }
}
\end{lstlisting}
\end{minipage}\begin{minipage}[t]{.53\textwidth}
\begin{lstlisting}
class Add extends Exp {
  Exp expr1;
  Exp expr2;
  Add setAdd(Exp a, Exp b) { expr1 = a; expr2 = b; return this; }
  String toString() { return expr1.toString() + 
                               " + " + expr2.toString(); }
}
\end{lstlisting}
\end{minipage}\vspace*{-1em}
\caption{Base Program}\label{fig:DLitAddNegPrint}
\end{figure}
This program comprises the class $\name{Exp}$, the class $\name{Lit}$ for literal expressions and the class $\name{Add}$ for addition expressions.
All these classes implement the $\name{toString}$ method.
Moreover, $\name{Lit}$ and $\name{Add}$ also have a setter method.

\paragraph{Implementing Feature Neg.}
%Delta modules can add or remove classes and attributes (method and fields) and modify classes and methods.
 Figure~\ref{fig:deltaNeg} presents the three delta modules  (introduced by the keyword $\kwdelta$)
that add the feature \textsf{Neg} to the base program.
\begin{figure}
\begin{minipage}[t]{.49\textwidth}\begin{lstlisting}
delta DNeg  {
   adds class Neg extends Exp {
      Exp expr;
      Neg setNeg(Exp a) { expr = a; return this; }
   }
}
delta DNegPrint {
   modifies Neg {
      adds String toString() {
       return "-" +  expr.toString(); }
   }
}
\end{lstlisting}\end{minipage}\begin{minipage}[t]{.49\textwidth}\begin{lstlisting}
delta DOptionalPrint {
   modifies Add {
      modifies String toString() {
       return "(" +  original() + ")"; }
   }
}
\end{lstlisting}\end{minipage}\vspace*{-1em}
\caption{Delta Modules for the \textsf{Neg} Feature}\label{fig:deltaNeg}
\end{figure}
Namely:  $\name{DNeg}$ adds the class \texttt{Neg} with a simple setter;
 $\name{DNegPrint}$ adds to class \texttt{Neg} the $\name{toString}$ method (relevant for the \textsf{Print} feature);
  and 
$\name{DOptionalPrint}$ 
adds glue code to ensure  that the two optional features  \textsf{Add} and \textsf{Neg} cooperate properly: 
it {\em modifies} the implementation of the $\name{toString}$ method of the class $\name{Add}$ by putting parentheses around the textual representation of a sum expression,
 thus avoiding ambiguity in printing.
E.g., without applying $\name{DOptionalPrint}$ both the following expressions 
\begin{lstlisting}
(new Add()).setAdd( new (Neg()).setNeg((new  Lit()).setLit(3)),  new (Lit()).setLit(5) )     //  (-3) + 5
(new Neg()).setNeg( new (Add()).setAdd((new  Lit()).setLit(3), new (Lit()).setLit(5)) )     //  -(3+5)
\end{lstlisting}
 would be printed as ``-3+5''; while after applying $\name{DOptionalPrint}$ the former is printed as ``(-3+5)'' and the latter is printed as ``-(3+5)''.
Delta module $\name{DOptionalPrint}$ illustrates the usage of the special method $\kworiginal$ which allows here to call the original implementation of the method $\name{toString}$,
  and surround the resulting string with parenthesis.

\paragraph{Implementing Features Eval1 and Eval2.}

Figure~\ref{fig:deltaEval} presents the delta modules 
 that add the features \textsf{Eval1} and \textsf{Eval2} (on the left and on the right, respectively).
\begin{figure}
\begin{minipage}[t]{.46\textwidth}\begin{lstlisting}
delta DLitEval1 {
   modifies Exp {
      adds int eval() { return 0; }
   }
   modifies Lit {
      adds int eval() { return value; }
   }
}
delta DAddEval1 {
   modifies Add {
      adds int eval() {
        return expr1.eval() + expr2.eval();
      }
   }
}
delta DNegEval1{
   modifies Neg {
      adds int eval() { return (-1) * expr.eval(); }
   }
}
\end{lstlisting}\end{minipage}\begin{minipage}[t]{.53\textwidth}\begin{lstlisting}
delta DLitEval2 {
   modifies Exp {
      adds Lit eval() { return null; }
   }
   modifies Lit {
      adds Lit eval() { return this; }
   }
}
delta DAddEval2 {
   modifies Add {
      adds Lit eval() {
        Lit res = exp1.eval(); 
        return res.setLit(res.value + exp2.eval()); }
   }
}
delta DNegEval2{
   modifies Neg {
      adds Lit eval() { Lit res = exp.eval();
        return res.setLit((-1) * res.value); }
   }
}
\end{lstlisting}\end{minipage}\vspace*{-1em}
\caption{Delta Modules for Features \textsf{Eval1} (left) and \textsf{Eval2} (right)}\label{fig:deltaEval}
\end{figure}
The delta module $\name{DLitEval1}$ (resp.\ $\name{DLitEval2}$) modifies the classes $\name{Exp}$ and $\name{Lit}$
  by adding to them the $\name{eval}$ method corresponding to the \textsf{Eval1} (resp.\ \textsf{Eval2}) feature:
  $\name{eval}$ takes no parameter and returns an int (resp.\ a {\tt Lit} object).
The delta module $\name{DAddEval1}$ (resp.\ $\name{DAddEval2}$) does the same operation on the $\name{Add}$ class;
 and the delta module $\name{DNegval1}$ (resp.\ $\name{DANegEval2}$) does the same operation on the $\name{Neg}$ class.

\paragraph{Removing the Add Feature.}
If the feature \textsf{Add} is not selected, the generated variant must not contain the class $\name{Add}$.
This is ensured by the delta module $\name{DremAdd}$ in Figure~\ref{fig:deltaRemAdd} which removes the class $\name{Add}$ from the program.
\begin{figure}
\begin{lstlisting}
delta DremAdd { removes Add }
\end{lstlisting}\vspace*{-1em}
\caption{Delta Module for Removing the \textsf{Add} Feature} \label{fig:deltaRemAdd}
\end{figure}

\subsection{The Configuration Knowledge}

The configuration knowledge specifies how variants are generated by
  i) specifying for which product (i.e., set of selected features) each delta module is activated,
  and ii) specifying a partial application order on the delta modules.
Figure~\ref{fig:EPL:CK} presents the activation conditions and the partial order of the delta modules.
\begin{figure}
{\scriptsize\begin{center}
\begin{tabular}{l}
  \begin{tabular}{@{}lc@{\qquad}c@{\qquad}c}
    Activations: & \begin{tabular}{@{}c@{}}\begin{tabular}{|l|l|}
    \hhline{+-|-+}
    {\bf Delta Module} & Activation \\
   \hhline{+-|-+}
    $\name{DNeg}$      & \textsf{Neg}                    \\
    $\name{DNegPrint}$ &\textsf{Neg} $\land$ \textsf{Print}     \\
    $\name{DOptionalPrint}$ & \textsf{Neg} $\land$ \textsf{Add} \\
    \hhline{+-|-+}
    \end{tabular}\\~\end{tabular} & \begin{tabular}{@{}c@{}}
    \begin{tabular}{|l|l|}
   \hhline{+-|-+}
    {\bf Delta Module} & Activation \\
    \hhline{+-|-+}
    $\name{DLitEval1}$ & \textsf{Eval1}                  \\
    $\name{DAddEval1}$ & \textsf{Eval1} $\land$ \textsf{Add}    \\
    $\name{DNegEval1}$ & \textsf{Neg} $\land$ \textsf{Eval1}    \\
    \hhline{+-|-+}
    \end{tabular}\\~\end{tabular} & \begin{tabular}{|l|l|}
    \hhline{+-|-+}
    {\bf Delta Module} & Activation \\
    \hhline{+-|-+}
    $\name{DLitEval2}$ & \textsf{Eval2}                  \\
    $\name{DAddEval2}$ & \textsf{Eval2} $\land$ \textsf{Add}    \\
    $\name{DNegEval2}$ & \textsf{Neg} $\land$ \textsf{Eval2}    \\
    $\name{DremAdd}$   & $\neg$\textsf{Add}              \\
    \hhline{+-|-+}
    \end{tabular}
  \end{tabular}\\[2em]
  Order: \begin{tabular}{l}
          $\name{DNeg} \isbefore \{\ \name{DNegPrint},\ \name{DOptionalPrint}\ \}$\\
           $\qquad\qquad \isbefore \{\ \name{DLitEval1},\ \name{DAddEval1},\ \name{DNegEval1}\  \} 
              \isbefore \{\ \name{DLitEval2},\ \name{DAddEval2},\ \name{DNegEval2}\  \}\isbefore \name{DremAdd}$
        \end{tabular}
\end{tabular}
\end{center}}
\vspace*{-1.2em}
\caption{Expression Product Line: Configuration Knowledge}\label{fig:EPL:CK}
\end{figure}
The activation conditions and the partial order reflect the  explanations about the delta modules of the \EPL\ given in Section~\ref{sec:example-AB}. For instance, the delta module {\tt DNeg} is activated whenever the feature \textsf{Neg} is activated, 
the delta module {\tt DremAdd} is activated whenever the feature \textsf{Add} is not selected, and the delta module {\tt DOptionalPrint} is activated
whenever both features  \textsf{Add} and \textsf{Neg} are activated (recall that feature \textsf{Print} is mandatory).

Following~\cite{BetDamSch:ACTA-2013}, the partial order is specified as a total order on a partition of the set of delta modules.
The partial order  must ensure that the variants of the \EPL\ can be generated.
  Therefore, it states that the delta modules  {\tt DNeg} (that adds the class {\tt Neg})
 must be applied before {\tt DNegPrint}, {\tt DNegEval1} and {\tt DNegEval2} (that modify class {\tt Neg}).
The partial order also ensures that, independently from the activation conditions, the delta modules occurring in the same partition perform disjoint delta operations (thus guranteeing that applying any subset of them in any possible order always produces the same transformation)---this guarantees that the product line is unambiguous  (i.e., applying  the activated delta modules in any possible total order that respects the application order produces the same variant).
Therefore, the delta modules for feature \textsf{Eval1} and the delta modules for feature  \textsf{Eval2} are put in two different parts;
 and the delta module {\tt DremAdd} (that removes the class {\tt Add}) is applied after  {\tt DAddEval1}, {\tt DAddEval2} and {\tt DOptionalPrint} (that modify class {\tt Add}).

\section{The \IFDJ\ Calculus}\label{sec:model}

\iffalse
TODO
 - add definition of monotonicity
 - show why it is an interesting property, on the example (more readable)
 - add notations on the order on delta modules (for the algorithm: notion of order on set of delta, with the normal set definitions on programs)
 - add notations for the operations (so we can have refactor generics)
   - Module(OP) = all the modules that perform that operation
   - \neg OP = the opposite operation (removes <-> adds)
 - add semantics of DOP (so we can prove that the semantics of the program does not change)
   - semantics of operations
   - semantics of feature selection
   - definition of strong unambiguity based on this semantics. Hmm, is this necessary? like the definition of errors: is that necessary?
   - hence, showing that the semantics is preserved is simply showing that one step of refactoring generates a selection set included in the selection set of the original program.
 - remove definitions for the type system: we don't use the same notations here
 - rewrite the example
\fi

In this section we briefly recall the \IFDJ~\cite{BetDamSch:ACTA-2013} core calculus for \DOP.
We present the calculus in two steps:
 (i) we introduce the \IFJ\ calculus, which is an imperative version of \FJ~\cite{IgarashiPierceWadler:TOPLAS-2001}; and
 (ii) we introduce the constructs for variability on top of it.
The full descriprion of  \IFDJ\ is given in~\cite{BetDamSch:ACTA-2013}, where a type-cheching technique for ensuring type soundness of all variants is presented.
The version of \IFDJ\ presented in this paper is indeed a slight extension of the one presented in~\cite{BetDamSch:ACTA-2013}:
 the \AB\ contains also an \IFJ\ program outside of any delta module. This makes the \IFDJ\ syntax a direct extension of the \IFJ\ syntax.
% thus simplifying the technical formulation of some results.

%\subsection{A recollection of \IFJ}

The abstract syntax of \IFJ\  is presented in Figure~\ref{fig:IFJ:syntax} (top).
Following~\cite{IgarashiPierceWadler:TOPLAS-2001}, we use the overline notation for (possibly empty) sequences of elements:
 for instance $\vect{\see}$ stands for a sequence of expressions.
Variables $\namex$ include the special variable $\kwthis$  (implicitly bound in any method declaration $\seMD$),
 which may not be used as the name of a formal parameter of a method. 
\begin{figure}[t]
%\vspace*{-0.7em}
\begin{syntax}[desc={right,flushright},descsep=.6em,size=\small]
\entry \sePIFJ = \vect{\seCD} [Program]
\entry \seCD   = \kwclass\ \nameC\ \kwextends\ \nameC\ \{\ \vect{\seAD}\ \} [Class]
\entry \seAD   = \seFD \bnfor \seMD  [Attribute (Field or Method)]
\entry \seFD   = \nameC\ \namef [Field]
\entry \seMD   = \nameC\ \namem(\vect{\nameC\ \namex})\ \{ \kwreturn\ \see; \} [Method]
\entry \see    = \namex \bnfor \see.\namef \bnfor \see.\namem(\vect{\see}) \bnfor \kwnew\ \nameC() \bnfor (\nameC)\see
                 \bnfor \see.\namef = \see \bnfor \kwnull %\bnfor\kwthis 
                 \hspace*{-3.3em} [Expression]
\end{syntax}%\vspace*{-1.5em}
%\caption{\IFJ\ Language: Syntax}\label{fig:IFJ:syntax}
%\end{figure}
%\begin{figure}[t]
%\vspace*{-0.7em}
\hrule
\vspace*{0.4em}
\begin{syntax}[desc={right,flushright},descsep=4em,size=\small]
\entry \sePIFDJ    =  \sePIFJ\quad \vect{\seDMD}\quad \seFM\quad\seCK [Product Line] %[Delta Program]
%\entry \seFormula  = \true \bnfor \namefeature \bnfor \seFormula \Rightarrow \seFormula \bnfor \neg \seFormula
%                     \bnfor \seFormula \land \seFormula \bnfor \seFormula \lor \seFormula [Formula]
\entry \seDMD      = \kwdelta\ \named\ \{\ \vect{\seCO}\ \} [Delta Module]
\entry \seCO       = \kwadds\ \seCD \bnfor \kwremoves\ \nameC \bnfor \kwmodifies\ \nameC\ {[ \kwextending\ \nameC']}\ \{\ \vect{\seAO}\ \}\hspace*{-3.1em} [Class Operation]
\entry \seAO       = \kwadds\ \seAD \bnfor \kwmodifies\ \seMD \bnfor \kwremoves\ \namea [Attribute Operation]
\end{syntax}\vspace*{-1.5em}
\caption{Synax of \IFJ\ (top) and \IFDJ\ (bottom)}\label{fig:IFDJ:syntax}\label{fig:IFJ:syntax}
\end{figure}
A program $\sePIFJ$ is a sequence of class declarations $\vect{\seCD}$.
A class declaration $\kwclass\ \nameC\ \kwextends\ \nameC'\ \{\ \vect{\seAD}\ \}$
comprises the name $\nameC$ of the class, the name $\nameC'$ of the superclass (which must always be specified, even if it is the built-in class $\name{Object}$),
 and a list of field and method declarations $ \vect{\seAD}$.
All fields and methods are public, there is no field shadowing, there is no method overloading,
 and each class is assumed to have an implicit constructor that initializes all fields to $\kwnull$.
The subtyping relation $\issubtof$ on classes, %(types),
  which is the reflexive and transitive closure of the immediate subclass relation (given by the $\kwextends$ clauses in class declarations), is supposed to be acyclic.
%Type system, operational semantics, and type soundness for \IFJ\ are given in~\cite{BetDamSch:ACTA-2013}.

%\subsection{\IFDJ\ as an Extension of \IFJ}\label{sec:IFDJ}

The abstract syntax of the language \IFDJ\ is given in Figure~\ref{fig:IFDJ:syntax} (bottom).
An \IFDJ\ \SPL\ $\sePIFDJ$ comprises: a possibly empty or incomplete  \IFJ\ program \sePIFJ;
 a set of delta modules $\vect{\seDMD}$ that, together with the base program  $\sePIFJ$, represents the artifact base;
 a feature model $\seFM$ specifying the features and the products of the \SPL;
 and a configuration knowledge $\seCK$ (i.e., the ordering between delta modules and their activation conditions).

To simplify the presentation, we do not give a syntactic description of $\seFM$ nor of $\seCK$ and we rely on getter functions as follows: 
  $\sePIFDJ.\gfeatures$ is the set of features; %(features are ranged over by $\namefeature$);
  $\sePIFDJ.\gformula$ specifies the products (i.e., a subset of the power set $2^{\sePIFDJ.\gfeatures}$);
  $\sePIFDJ.\gdeltacondition$ maps each delta module name $\named$ to its activation condition;
and $\sePIFDJ.\gorder$ (or $\isbefore$, for short) is the application ordering between the delta modules.

A delta module declaration $\seDMD$ comprises the name $\named$ of the delta module
 and class operations $\vect{\seCO}$ representing the transformations performed when the delta module is applied to an \IFJ\ program.
A class operation can add, remove, or modify a class.
A class can be modified by (possibly) changing its super class and performing  attribute operations $\vect{\seAO}$  on its body.
An {\em attribute name} $\namea$  is either a field name $\namef$ or a method name $\namem$.
An attribute operation can add or remove fields and methods, and modify the implementation of a method by replacing its body. 
The new body may call the special method $\kworiginal$,
 which is implicitly bound to the previous implementation of the method and may not be used as the name of a method. 
%The class operations in a delta module must act on distinct classes, and the attribute operations in a class operation must act on distinct attributes.

% here was the example previously

The \emph{projection} of a product line on a subset of its products is obtained by
 restricting  $\sePIFDJ.\gformula$ to describe only the products in the subset and by dropping delta modules that are never activated.

\begin{example}\label{exa:prj-Eval}
For instance, the \AB\ of the projection of the \EPL\ on the products without feature \textsf{Neg}
 is obtained by dropping the delta modules \name{DNeg}, \name{DNegPrint} and \name{DOptionalPrint};
 and the \AB\ of the projection of the \EPL\ on the products without feature \textsf{Eval2}
 is obtained by dropping the delta modules \name{DLitEval2}, \name{DAddEval2} and \name{DNegEval2}. 
\end{example}

%\section{Model}\label{sec:model}\input{src/30-model}
\section{Auxiliary Notations}\label{sec:auxdef}

In this section we introduce some auxiliary notations that will be used in Section~\ref{sec:algo}.
Our first notation relates the $\kwmodifies$ operators on methods to the concept of monotonicity.
Indeed, in general $\kwmodifies$ on methods is not monotonic:
  the body of the method is replaced by some code that can be entirely different.
However, we can distinguish two cases in which $\kwmodifies$ can be considered monotonic:
 when it calls $\kworiginal$, the generated variant contains the original body of the method, and so $\kwmodifies$ can be considered {\em increasing monotonic};
 when the body of the method is {\em voided} (i.e., it is replaced by $\kwreturn\ \kwnull$) $\kwmodifies$ can be considered {\em decreasing monotonic}.
\begin{notation}[$\kwwraps$ and $\kwvoids$]\label{def:wrapsANDreplaces}
Let $\kwwraps$ denote a $\kwmodifies$ operation on method that calls $\kworiginal$, and $\kwvoids$ denote a $\kwmodifies$ operation that removes the content of a method:
 $\kwvoids\ \namem$ corresponds to $\kwmodifies\ \nameC\ \namem(\cdots)\ \{\kwreturn\; \kwnull\}$.
\end{notation}

The goal of the two following notations is to unify delta operations on classes and on attributes in a single model,
  in order to manage uniformly these two kind of operations in our refactoring algorithms.
Using these notations simplifies the description of our refactoring algorithms.
\begin{notation}
A {\em reference}, written $\nameel$, is either a class name $\nameC$ or a qualified attribute name $\nameC.\namea$
 and we write $\nameel\leq\nameel'$ if $\nameel=\nameel'$ or if $\nameel$ is a prefix of $\nameel'$.
By abuse of notation, we also consider the $\kwextends$ clause as an attribute of its class, and consider $\nameC.\kwextends$ as a valid reference.
\end{notation}

\begin{notation}\label{def:ADO}
We abstract a delta module by a set of {\em Abstract Delta Operations} (\ADO) which are triplets $(\op, \nameel, \data)$ where:
 i) $\op$ is a delta operation keyword ($\kwadds$, $\kwremoves$ or $\kwmodifies$),
 ii) $\nameel$ is the reference on which $\op$ is applied,
 iii) $\data$ is the data associated with this operations, and
 iv) if $\op=\kwmodifies$ then  $\nameel$ is not a class name. 
%\begin{enumerate}
%\item $\op$ is a delta operation keyword ($\kwadds$, $\kwremoves$ or $\kwmodifies$),
%\item $\nameel$ is the reference on which $\op$ is applied,
%\item $\data$ is the data associated with this operations, and
%\item if $\op=\kwmodifies$ then  $\nameel$ is not a class name. 
%\end{enumerate}
Given an \ADO\ $\ado$, we denote its operator as $\ado.\op$, its reference as $\ado.\nameel$ and its data as $\ado.\data$.
\end{notation}
These two notations are ilustrated by the following examples.  
In particular, the first example shows that a $\kwmodifies$ operation on a class $\nameC$
 that contains only $\kwadds$ operations on attributes is represented by the set of \ADO s containing only the $\kwadds$ operations:
 the $\kwmodifies\ \nameC$ operation is only a syntactic construction to introduce these $\kwadds$ operations and is not included in our representation.

\begin{example}
The delta module \name{DLitEval2} in Figure~\ref{fig:deltaEval} that modifies classes \name{Exp} and \name{Lit} by adding a method \name{eval} to each of them,
 is modeled with only two \ADO s:
 \begin{center}
  (\lstinline[basicstyle=\rm\sffamily\footnotesize]|adds, $\;\;$ Exp.eval, $\;\;$ Lit eval() { return null; }|)
   $\quad$ and $\quad$ (\lstinline[basicstyle=\rm\sffamily\footnotesize]|adds, $\;\;$ Lit.eval, $\;\;$ Lit eval() { return this; }|)
 \end{center}
These \ADO s model the addition of the {\tt eval} methods,
 the modification of classes \name{Exp} and \name{Lit} being implicit as \name{Exp} (resp. \name{Lit}) is a prefix of $\name{Exp.eval}$   (resp.\ $\name{Lit.eval}$).
\end{example}

\begin{example}
The delta module \name{DOptionalPrint} in Figure~\ref{fig:deltaNeg} that modifies the class \name{Add} by modifying the method \name{toString}, is modeled with only one \ADO:
\begin{center}
(\lstinline[basicstyle=\rm\sffamily\footnotesize]|modifies, $\;\;$ Add.toString, $\;\;$ String toString() { return ``('' + original() + ``)''; }|)
\end{center}
\end{example}

%\begin{example}
%A delta module that modifies a class $\nameC$ by adding to it a method $\seMD$ named $\namem$ and  by changing its superclass to $\nameC'$ is modeled by two \ADO s:
%  $(\kwadds, \nameC.\namem, \seMD)$ for the addition of the method, and  $(\kwmodifies, \nameC.\kwextends, \nameC')$ for the modification of the $\kwextends$ clause.
%\end{example}

\begin{example}\label{exa:proj-NoNeg}
Note  that, according to Definition~\ref{def:ADO}, the projection of the \EPL\ on the products without feature \textsf{Neg} does not contain $\kwmodifies$ operations.
\end{example}

Our last notations are used to iterate over delta modules:
 first, we present the notations to get a set of delta module names,
 then we present the notations to order such a set so to iterate over it in a \lstinline[language=Algorithm,basicstyle=\normalsize]|for| loop.

\begin{notation}
The set of delta module names declared in $\sePIFDJ$ is denoted as $\gdmodule(\sePIFDJ)$.
When $\sePIFDJ$ is clear from the context, we write $\fbefore(\named)$ the set of delta module names that are before $\named$ for $\sePIFDJ.\gorder$.
\end{notation}

\begin{notation}
Given a set of delta names $S=\{\named_i\sht i\in I\}$, we denote $\orderup{S}$ (resp. $\orderdown{S}$) a sequence $(\named_{i_1},\dots,\named_{i_n})$
 of all the names in $S$ that respects the partial order (resp. the partial order opposite from the one) specified by $\sePIFDJ.\gorder$.
\end{notation}

\section{Monotonicity and Refactoring Algorithms}\label{sec:algo}
%Inutitively, the notion of Monotonicity is quite simple:
%Informally,  {\em increasing monotonicity} is constructing a variant only by  adding  new content to a base program,
%  while {\em decreasing monotonicity} is constructing it only by   removing  content from the base program.
In the introduction, we pointed out that the flexibility provided by delta operations, being very useful for easily constructing \SPL s,
 can lead to unnecessary complexity with many adding and removing operations cancelling each other.
Monotonicity is a natural approach to lower such complexity as it forbids opposite adding and removing operations:
 informally, {\em increasing monotonicity} is constructing a variant only by adding new content to the base program
  and is in principle similar to {\em Feature-Oriented Programming} (\FOP)~\cite{Batory:2003};\footnote{As pointed out in~\cite{FOSD}, \DOP\ is a generalization of \FOP:
 the \AB\ of a \FOP\ product line consists of a set of \emph{feature modules} which are delta modules that correspond one-to-one to features and do not contain remove operations.}
 on the other hand, {\em decreasing monotonicity} is constructing it only by removing content from the base program
  and share similarities with annotative  approaches (see, e.g.,~\cite{czarnecki2005,Kastner-EtAl:ACM-TSE-2012}).

%However, such a simple notion has two major flaws:
% first, it limitates code reuse as it is impossible for a delta module to {\em wrap} or {\em replace} methods that have been introduced in the base program or in other delta modules;
% second, it is in general difficult and even in some cases impossible to refactor an \SPL\ into an equivalent \SPL\ following such definitions.

Section~\ref{sec:algo-inc} focuses on increasing monotonicity:
 it formalizes and motivates different levels of purity for it,
 then presents a refactoring algorithm transforming an \SPL\ into an increasing monotonic equivalent
 and illustrates it on the \EPL\ example.
Section~\ref{sec:algo-dec} formalizes decreasing monotonicity, presents a refactoring algorithm and its application to the \EPL.
Section~\ref{sec:algo-prop} gives correctness and complexity of the refactoring algorithms.

\subsection{Increasing Monotonicity}\label{sec:algo-inc}

%As we previously discussed, a 
%A first intuitive notion of increasing monotonicity is only to allow $\kwadds$ operations:

Before presenting the first refactoring algorithm, we gradually introduce three notions of increasing monotonicity,
 from the most intuitive one, called {\em strictly-increasing}, to the most flexible one, called {\em pseudo-increasing}.
Depending on the properties of the input \SPL, the algorithm can produce \SPL s corresponding to any of the three notions.
A first intuitive notion of increasing monotonicity is only to allow $\kwadds$ operations:
\begin{definition}[Strictly-increasing monotonic]\label{def:strict-im}
An \SPL\ is {\em strictly-increasing monotonic} iff it only contains $\kwadds$ operations.
\end{definition}
\noindent
Note that this notion is quite restrictive, as it does not allow the extension of method implementation, or the modification of the $\kwextends$ clause of a class,
  two operations possible in \FOP.
%Note that the above notion does not encompass \emph{Feature-Oriented Programming} (\FOP)~\cite{Batory:2003}
%and is not encompassed by \FOP.\footnote{\DOP\ is a generalization of \FOP:
% the \AB\ of a \FOP\ product line consist of a set of \emph{feature modules} which are delta modules that correspond one-to-one to features and do not contain remove operations.}
%
%Note that the above notion does not encompass \emph{Feature-Oriented Programming} (\FOP)~\cite{Batory:2003}
%and is not encompassed by \FOP.\footnote{\DOP\ is a generalization of \FOP:
% the \AB\ of a \FOP\ product line consist of a set of \emph{feature modules} which are delta modules that correspond one-to-one to features and do not contain remove operations.}
%However, such a simple notion  has a major drawback: since it rules out the $\kwmodifies$ operation, in general to refactor an \SPL\ to follow Definition~ \ref{def:strict-im} may introduces a considerable amount of code duplication. \begin{TODO}Ci vorrebbe un esempio.\end{TODO}
%
The following more liberal notion allows to increase the body of existing methods by using the $\kwmodifies$ operator by always calling $\kworiginal$.
Still, it  does not include the modification of the $\kwextends$ clause of a class present in \FOP.
%to {\em wrap} its previous implementation with new code:
%\begin{definition}[$\kwwraps$ and  $\kwreplaces$]\label{def:wraps-replaces}
%Let $\kwwraps$ denote a $\kwmodifies$ operation that calls $\kworiginal$, and $\kwreplaces$ a $\kwmodifies$ operation that does not.
%\end{definition}
\begin{definition}[Increasing Monotonic]\label{def:im}
An \SPL\ is {\em increasing monotonic} iff it only contains $\kwadds$ and $\kwwraps$ operations.
\end{definition}

\noindent
The last notion, which is a generalization of \FOP, is to allow $\kwmodifies$ also to modify the $\kwextends$ clause of a class and to replace the implementation of a method,
 leaving only $\kwremoves$ as a forbidden operation:
\begin{definition}[Pseudo-increasing monotonic]\label{def:pseudo-im}
An \SPL\ is {\em pseudo-increasing monotonic} iff it does not contain $\kwremoves$ operations.
\end{definition}
\noindent
We have qualified the above notion as \emph{pseudo-}, since it allows delta modules to replace the $\kwextends$ clause of a class and to remove or entirely replace content from the body of method definitions.
Thus, it does not reflect the informal definition of increasing monotonicity given at the beginning of Section~\ref{sec:algo}.

\subsubsection{Increasing Monotonicity Refactoring Algorithm}\label{sec:algo-inc-algo}

\begin{figure}[t]%$\setcounter{lstnumber}{\numexpr\thelstnumber-1\relax}$
\begin{minipage}[t]{.49\textwidth}
\begin{lstlisting}[language=Algorithm,basicstyle=\footnotesize,numbers=left]
Delta Module Name: $\named_1$, $\named_2$;
Operation: $\ado_1$, $\ado_2$;
Set $\text{\bf of}$ Delta Module Name: $S$;

refactor($\sePIFDJ$) = 
  for $\named_1\in\orderup{\gdmodule(\sePIFDJ)}$ do
    for $\ado_1\in\sePIFDJ(\named_1)$ do
      if($\ado_1.\op=\kwremoves$)
        $\sePIFDJ(\named_1)$ := $\sePIFDJ(\named_1)\setminus\ado_1$
        manageOperation()
      fi
    done
  done;

manageOperation() =
  $S$ := $\emptyset$
  for $\named_2\in\orderdown{\fbefore(\named_1)}$ do
    for $\ado_2\in\sePIFDJ(\named_2)$ do
      if($\ado_1.\nameel\leq \ado_2.\nameel$) mergeOperations() fi
    done
  done
  mergeToBase();
\end{lstlisting}
\end{minipage}\begin{minipage}[t]{.50\textwidth}
\begin{lstlisting}[language=Algorithm,basicstyle=\footnotesize,numbers=left,firstnumber=last]
mergeOperations() =
  $S$ := $S\cup\{\named_2\}$
  $\sePIFDJ(\named_2)$ := $\sePIFDJ(\named_2)\setminus\ado_2$
  if($\sePIFDJ(\named_2)=\{\ \}$) $\sePIFDJ$ := $\sePIFDJ\setminus\named_2$ fi
  $\sePIFDJ$ := $\sePIFDJ + \named\fresh$ with {
    $\sePIFDJ(\named)$ := $\{\ \ado_2\ \}$
    $\sePIFDJ.\gdeltacondition(\named)$ := $\named_2\land\neg\named_1$
    $\sePIFDJ.\gorder(\named)$ := $\sePIFDJ\gorder(\named_2)$
  }

mergeToBase() =
  $\data$ := $\sePIFDJ.\sePIFJ(\ado_1.\nameel)$
  if($\data$ != $\bot$)
    $\sePIFDJ.\sePIFJ$ := $\fapply(\ado_1, \sePIFDJ.\sePIFJ)$
    $\sePIFDJ$ := $\sePIFDJ + \named\fresh$ with {
      $\sePIFDJ(\named)$ := $\{\ (\ \kwadds,\; \ado_1.\nameel,\; \data\ )\ \}$
      $\sePIFDJ.\gdeltacondition(\named)$ := $\neg\named_1$
      $\sePIFDJ.\gorder(\named)$ := before($S$)
    }
  fi;
\end{lstlisting}
\end{minipage}\vspace*{-1em}
\caption{Refactoring Algorithm for Increasing Monotonic \SPL}\label{fig:refactoring:increasing}\vspace*{-2em}
\end{figure}
%
%We focus on refactoring an \SPL\ only at the level of delta operations:
% the goal is to cancel all $\kwremoves$ operations from the input \SPL.
%The algorithm  refactors a \DOP\ product line $\sePIFDJ$ by removing all $\kwremoves$ operations and without removing or introducing new $\kwmodifies$ operations.
The refactoring algorithm, presented in Figure~\ref{fig:refactoring:increasing}, transforms its input \DOP\ product line $\sePIFDJ$
 by eliminating all $\kwremoves$ operations and without eliminating or introducing new $\kwmodifies$ operations.
Therefore, the refactored \SPL\ is \vspace*{-.3em}
\begin{itemize}
\item strictly-increasing, if $\sePIFDJ$ does not contain $\kwmodifies$ operations;\vspace*{-.3em}
%\item increasing, if $\sePIFDJ$ does not contain $\kwreplaces$ operations; and
\item increasing, if all the $\kwmodifies$ operations in $\sePIFDJ$ are $\kwwraps$ operations; and\vspace*{-.3em}
\item pseudo-increasing, otherwise.
\end{itemize}
Note that the algorithm may turn an existing delta module into an empty delta module
  which can then can be removed by a straightforward algorithm (see~\cite{Schulze:2013:RDS:2451436.2451446}).
%Empty delta modules can be removed from the refactored product line by  a straightforward algorithm (see~\cite{Schulze:2013:RDS:2451436.2451446}).

%As we discussed before, we focus in this paper in refactoring the \SPL\ only at the level of delta operations:
%  this means that the goal of this algorithm is to cancel all $\kwremoves$ operations from the input \SPL.
To illustrate how the refactoring algorithm works, consider a delta module $\named$ containing a removal operation on an element $\nameel$ (either a class or an attribute).
This operation would be applied only when $\named$ is activated,
 and would remove all declarations (and modification) of $\nameel$ that are done {\em before} the application of $\named$.
Hence, to cancel this removal operation, we can simply transform the \SPL\ so that $\nameel$ is never declared before $\named$ and when it is activated.

%
%The algorithm is presented in Figure~\ref{fig:refactoring:increasing} and is structured in four functions with four global variables.
The algorithm is structured in four functions with four global variables.
The main function of our algorithm is {\tt refactor} which takes the \SPL\ to refactor as parameter.
This function looks in order at all the delta modules and when finding a $\kwremoves$ operation $\ado_1$ inside a delta module $\named_1$,
 it cancels it from  $\named_1$ and calls the {\tt manageOperation} function.
The goal of the {\tt manageOperation} function is to transform the \SPL\ for the $\ado_1$ operation as described before.
It is structured in two parts.
First, it looks in order at all the delta operations applied before $\named_1$,
 and upon finding an operation $\ado_2$ in a delta module $\named_2$ that manipulates $\ado_1.\nameel$,
 it calls {\tt mergeOperation} which extracts that operation from $\named_2$ and changes the application condition of $\ado_2$ (using a freshly created delta module $\named$)
 so it is executed only when $\ado_1$ would not be executed.
Second, it calls {\tt mergeToBase} which looks if the element removed by $\ado_1$ is declared in the base program, 
 and if so, extracts it from the base program into a fresh opposite delta module $\named$ that is activated only when $\ado_1$ would not be executed.
The addition of this new delta module is done in lines 37--41 where we state that $\sePIFDJ$ is changed by adding a fresh delta module $\named$ with the following characteristics:
  its set of \ADO\ $\sePIFDJ(\named)$ is the singleton $(\kwadds, \ado_1.\nameel, \data)$ that adds $\ado_1.\nameel$ again to the base program;
  its activation condition $\sePIFDJ.\gdeltacondition(\named)$ is the opposite of $\named_1$;
  and its ordering $\sePIFDJ.\gorder(\named)$ states that it must be applied before all the delta modules in $S$.

There are three subtleties in this algorithm.
First, to deal with the fact that removing a class also removes all its attributes, the condition in line 19 is ``$\ado_1.\nameel\leq \ado_2.\nameel$''
 meaning that: if $\ado_1$ removes a class $\nameC$, then previous additions and modifications of $\nameC$ and its attributes will be changed with {\tt mergeOperation}.
Second, in line 26, empty delta modules are eliminated to avoid creating too much of them.
Third, we compute in $S$ the set of all delta modules manipulating $\ado_1.\nameel$ before $\named_1$
 to set the order relation of the delta module created in the {\tt mergeToBase} function.

\subsubsection{Example: Refactoring the \EPL\ into Increasing Monotonicity}

We applied our implementation of this algorithm on the \EPL\ given in Section~\ref{sec:example}.
It contains only one $\kwremoves$ operation, in the {\tt DremAdd} delta module, removing the \textsf{Add} class.
Thus, by construction of our algorithm,
 only the delta modules {\tt DAddEval1}, {\tt DAddEval2}, {\tt DOptionalPrint} and the base program, that modify and declare the \textsf{Add} class (respectively),
 are changed by the refactoring process.

Let us illustrate the modification done on the delta modules by considering {\tt DAddEval1}:
 the function {\tt mergeOperations} extract the only operation inside this delta module (line 25),
 removes {\tt DAddEval1} as it is now empty (line 26), and then basically recreates it (line 27),
 with the activation condition extended with $\neg\text{\tt DremAdd}$, corresponding to $\text{\tt Add}$. 
Hence, the delta modules are simply renamed by the algorithm.
However, the base program is changed by the function {\tt mergeToBase} which removes the class \textsf{Add} from it,
 and creates a new delta module reintroducing that class with the activation condition $\neg\text{\tt DremAdd}$ which corresponds to $\text{\tt Add}$.

The modified delta modules are shown in Figure~\ref{fig:epl:increasing}.
The modified base program, which is not shown, is obtained from the original base program (see Figure~\ref{fig:DLitAddNegPrint}) by dropping the declaration of class \name{Add}.
Note that, since all the $\kwmodifies$ operations of the original SPL were $\kwwraps$ operations,
 the refactored SPL is increasing monotonic. % (cf.\ the  discussion at the beginning of Section~\ref{sec:algo-inc-algo}).
%Note that, since the original SPL does not contain method $\kwreplaces$ operations,
% the refactored SPL is increasing monotonic (cf.\ the  discussion at the beginning of Section~\ref{sec:algo-inc-algo}).
On the other hand, since the projection of the original \EPL\ on the products without feature \textsf{Neg} does not contain $\kwmodifies$ operations (see Example~\ref{exa:proj-NoNeg} in Section~\ref{sec:auxdef}),
 its increasing monotonic refactoring would produce a strict-increasing product line.
\begin{figure}
\begin{minipage}[t]{.46\textwidth}
\begin{lstlisting}
delta DNotDremAdd {
  adds class Add extends Exp {
    Exp expr1;
    Exp expr2;
    Add setAdd(Exp a, Exp b) {
      expr1 = a; expr2 = b; return this; }
    String toString() { return expr1.toString()
      + " + " + expr2.toString(); }
} }
delta DOptionalPrint_DremAdd {
   modifies Add {
      modifies String toString() {
       return "(" +  original() + ")"; } 
} }
\end{lstlisting}
\end{minipage}\begin{minipage}[t]{.53\textwidth}
\begin{lstlisting}
delta DAddEval1_DremAdd {
   modifies Add {
      adds int eval() {
        return expr1.eval() + expr2.eval();
      }
} }
delta DAddEval2_DremAdd {
   modifies Add {
      adds Lit eval() {
        Lit res = exp1.eval(); 
        return res.setLit(res.value + exp2.eval()); }
} }
\end{lstlisting}
\end{minipage}\vspace*{-1em}
\caption{Delta Modules of the \EPL\ Changed by the Increasing Refactoring Algorithm}\label{fig:epl:increasing}
\end{figure}

\vspace*{-.3em}

\subsection{Decreasing Monotonicity}\label{sec:algo-dec}

Like for increasing monotonicity, we %first 
introduce several levels of purity for decreasing monotonicity before presenting the refactoring algorithm.
Straightforward adaptations of Definition~\ref{def:strict-im},~\ref{def:im} and~\ref{def:pseudo-im} lead to the following definitions of 
strictly-decreasing, decreasing and pseudo-decreasing monotonicity.

%A first intuitive notion of decreasing monotonicity is only to allow $\kwremoves$ operations:

\begin{definition}[Strictly-decreasing monotonic]\label{def:strict-dm}
An \SPL\ is {\em strictly-decreasing monotonic} iff it only contains $\kwremoves$ operations.
\end{definition}

\begin{definition}[Decreasing Monotonic]\label{def:dm}
An \SPL\ is {\em decreasing monotonic} iff it only contains  $\kwremoves$ operations and $\kwvoids$ operations.
\end{definition}

\begin{definition}[Pseudo-decreasing monotonic]\label{def:pseudo-dm}
An \SPL\ is {\em pseudo-decreasing monotonic} iff it only contains   $\kwremoves$ and $\kwmodifies$ operations.
\end{definition}

Unfortunately, the three above notions suffer of a major drawback: not all  product lines can be expressed by following their prescriptions.
For instance,  in order to conform to any of Definition~\ref{def:strict-dm},~\ref{def:dm} and~\ref{def:pseudo-dm},
 the base program of the EPL (cf.\ Section~\ref{sec:example}) must contain the class declaration
\begin{lstlisting}
class Exp extends Object {
     String toString() { return null; }
     Lit eval() { return null; }
     int eval() { return 0; }
}
\end{lstlisting}\vspace*{-.5em}
that contains two method declarations with same signature \textsf{eval()} and therefore is not valid in Java.
In order to overcome this drawback, we introduce the following notation to express the notion of ``readding''  (i.e., to remove and to immediately add)  an attribute.

\begin{notation}[$\kwreadds$]\label{def:readds}
Let $(\kwreadds,\nameel,\data)$ denotes the sequence of removing the attribute $\nameel$, and then performing $(\kwadds,\nameel,\data)$.
\end{notation}
\noindent We can now give the definitions of 
read-strictly-decreasing, readd-decreasing and read-pseudo-decreasing monotonicity that does not suffer of the above drawback.

\begin{definition}[Readd-strictly-decreasing monotonic]\label{def:readd-strict-dm}
An \SPL\ is {\em readd-strictly-decreasing monotonic} iff it only contains $\kwreadds$ and $\kwremoves$  operations.
\end{definition}

\begin{definition}[Readd-decreasing monotonic]\label{def:readd-dm}
An \SPL\ is {\em readd-decreasing monotonic} iff it only contains  $\kwreadds$ operations, $\kwremoves$ operations and  $\kwvoids$ operations.
\end{definition}

\begin{definition}[Readd-pseudo-decreasing monotonic]\label{def:readd-pseudo-dm}
An \SPL\ is {\em readd-pseudo-decreasing monotonic} iff it only contains $\kwreadds$, $\kwremoves$ and $\kwmodifies$ operations.
\end{definition}

\subsubsection{Decreasing Monotonicity Refactoring Algorithm}\label{sec:algo-dec-algo}

Our algorithm, presented in Figure~\ref{fig:refactoring:decreasing}, refactors a \DOP\ product line $\sePIFDJ$
 by eliminating  all $\kwadds$ operations and  without eliminating or introducing new $\kwmodifies$ operations.
Therefore, the refactored SPL  is \vspace*{-.3em}
\begin{itemize}
\item readd-strictly-decreasing if $\sePIFDJ$ does not contain $\kwmodifies$ operations;\vspace*{-.3em}
\item readd-decreasing if  all the $\kwmodifies$ operations in $\sePIFDJ$ are $\kwvoids$ operations; and\vspace*{-.3em}
\item readd-pseudo-decreasing, otherwise.
\end{itemize}
The decreasing monotonic refactoring algorithm may introduce empty new delta modules.
As pointed out in the discussion at the beginning of Section~\ref{sec:algo-inc-algo}, empty delta modules can be removed from the refactored product line by  a straightforward algorithm.
Moreover, if each class/attribute is introduced (i.e., either declared in the base program or added by a delta module) only once,
 then decreasing monotonic refactoring does not introduce $\kwreadds$ operations.

The structure of this refactoring algorithm is similar to the one to get increasing monotonicity:
  the main function {\tt refactor} takes as parameter the \SPL\ to refactor, and iterates over all the delta modules to find an $\kwadds$ operator to remove.
\begin{figure}
\begin{minipage}[t]{.53\textwidth}
\begin{lstlisting}[language=Algorithm,basicstyle=\footnotesize,numbers=left]
Delta Module Name $\named_1$, $\named_2$;
Operation $\ado_1$, $\ado_2$;

refactor($\sePIFDJ$) = 
  for module $\named_1\in\orderup{\gdmodule(\sePIFDJ)}$ do
    for $\ado_1\in\sePIFDJ(\named_1)$ do
      if($\ado_1.\op=\kwadds$)
        $\sePIFDJ(\named_1)$ := $\sePIFDJ(\named_1)\setminus\ado_1$
        manageOperation()
      fi
    done
  done;

manageOperation() =
  for module $\named_2\in\orderdown{\fbefore(\named_1)}$ do
    for $\ado_2\in\sePIFDJ(\named_2)$ do
      if(($\ado_2.\nameel\in\dom(\ado_1)$) & ($\ado_2.\op = \kwremoves$))
        mergeOperations()
      fi
    done
  done
  mergeToBase();
\end{lstlisting}
\end{minipage}\begin{minipage}[t]{.46\textwidth}
\begin{lstlisting}[language=Algorithm,basicstyle=\footnotesize,numbers=left,firstnumber=last]
mergeOperations() =
  $\sePIFDJ(\named_2)$ := $\sePIFDJ(\named_2)\setminus\ado_2$
  if($\sePIFDJ(\named_2)=\emptyset$) $\sePIFDJ$ := $\sePIFDJ\setminus\named_2$ fi
  $\sePIFDJ$ := $\sePIFDJ + \named\fresh$ with {
    $\sePIFDJ(\named)$ := $\{\ \ado_2\ \}$
    $\sePIFDJ.\gdeltacondition(\named)$ := $\named_2\land\neg\named_1$
    $\sePIFDJ.\gorder(\named)$ := $\sePIFDJ.\gorder(\named_2)$
  };

mergeToBase() =
  Set $\text{\bf of reference}$: $S$ := $\dom(\sePIFDJ.\sePIFJ)$
  $\sePIFDJ.\sePIFJ$ := $\sePIFDJ.\sePIFJ\cup\{\nameel\ \data \sht (\kwadds,\nameel,\data)\in\ado_1\land\nameel\not\in S\}$;
  $\sePIFDJ$ := $\sePIFDJ + \named\fresh$ with {
    $\sePIFDJ(\named)$ := $\{\ (\kwreadds,\nameC.\namea,\ado_1(\nameel)) \sht \nameC.\namea\in\dom(\ado_1)\cap S \}$
    $\sePIFDJ.\gdeltacondition(\named)$ := $\named_1$
    $\sePIFDJ.\gorder(\named)$ := $\sePIFDJ.\gorder(\named_1)$
  } $+$ $\named'\fresh$ with {
    $\sePIFDJ(\named')$ := $\{\ (\kwremoves, \nameel, \emptyset) \sht \nameel\in\dom(\ado_1)\setminus S \}$
    $\sePIFDJ.\gdeltacondition(\named')$ := $\neg\named_1$
    $\sePIFDJ.\gorder(\named')$ := $\sePIFDJ.\gorder(\named_1)$
  };
\end{lstlisting}
\end{minipage}\vspace*{-1em}
\caption{Refactoring Algorithm for Decreasing Monotonic \SPL}\label{fig:refactoring:decreasing}\vspace*{-.5em}
\end{figure}
Upon finding an operation $\ado_1$ with an $\kwadds$ operator in a delta module $\named_1$, the function {\tt manageOperation} is called.
This function, like for the increasing refactoring algorithm, is structured in two parts.
First, it looks in order at all the delta operations applied before $\named_1$,
 and upon finding an operation $\ado_2$ in a delta module $\named_2$ that manipulates $\ado_1.\nameel$ with a $\kwremoves$ operator,
 it calls {\tt mergeOperation} which extracts that operation from $\named_2$ and update the application condition of $\ado_2$ as done in the other algorithm.
Second, it calls {\tt mergeToBase} which integrates the operations $\ado_1$ in the base program as follows:
 first, it completes the base program with all the declarations introduced in $\ado_1$ that was missing from it;
 second, it creates a new delta module $\named$ that readds (see Definition~\ref{def:readds}) all the declarations originally done in the base program by the ones done in $\ado_1$;
 finally, it creates a new delta module $\named'$ opposite to $\ado_1$ that removes all the declarations done in $\ado_1$ if these operations would not be executed.
For the creation of these delta modules in lines 35--43, we use the following notations:
  $\dom(\ado)$ is the set of references that are declared in that operations,
  and $\ado(\nameel)$ is the data $\data$ associated to $\nameel$ in $\ado$.
For instance, with $\ado$ being the $\kwadds$ operation in the {\tt DNeg} delta module, we have
\begin{center}{\footnotesize
 $\dom(\ado)=\{\name{Neg},\name{Neg.expr},\name{Neg.setNeg}\}\qquad \text{\normalsize and, e.g.,}\qquad \ado(\name{Neg.expr})=(\name{Exp}\ \name{expr})$
}\end{center}
There are two subtleties in this algorithm.
First, it can occur that before an $\kwadds$ operation adding a class $\nameC$, removal operations can be applied on the {\em attributes} of $\nameC$,
 and so, the condition in line 17 ``$\ado_2.\nameel\in\dom(\ado_1)$'' captures all possible attributes of $\ado_1.\nameel$.
Second, in line 36, we only readd attributes, not classes, to ensure that the base program contains every elements declared in the \SPL.
Note also that in this example, there is no need of a set $S$ to define the order of the delta modules created in {\tt mergeToBase}:
  the order simply is the one of the original $\named_1$ delta module.

\subsubsection{Example: Refactoring the \EPL\ into Decreasing Monotonicity}

We applied this refactoring algorithm to the \EPL\ example.
All its delta modules but {\tt DremAdd} and {\tt DOptionalPrint} add new content to the base program, and all of them are modified by the refactoring as follows:
  they are emptied out by the {\tt refactor} function which removes the $\kwadds$ operations,
  that are then reintroduced to the \SPL\ by the {\tt mergeToBase} in the base program with few new delta modules.
The structure of the resulting \SPL\ is presented in Figure~\ref{fig:epl:decreasing}---it contains 8 empty delta modules
 (lines 27, 29, 31, 33, 38, 41, 44 and 47), which can be straightforwardly removed.
The left part of Figure~\ref{fig:epl:decreasing} contains the new base program which now contains all the elements declared in the \SPL:
  the class \textsf{Neg} as well as the attributes {\tt toString} and {\tt eval} are declared in the base program.
Note that as the delta modules implementing the \textsf{Eval1} feature are before the ones implementing the \textsf{Eval2} feature,
  the new base program contains the \textsf{Eval1} version of the {\tt eval} methods.
The right part of Figure~\ref{fig:epl:decreasing} presents the newly added delta modules.
The names of these delta modules are constructed in two parts: first the operation they perform, and then the delta module that created them.
For instance, {\tt DremNeg\_DNeg} is the removing delta module created in the {\tt mergeToBase} function from the {\tt DNeg} delta module:
 it removes the \textsf{Neg} class when the feature \textsf{Neg} is not selected.
The second delta module {\tt DremNegToString\_DNegPrint} is the delta module removing the method $\name{Neg}.\name{toString}$
 when neither \textsf{Neg} nor \textsf{Print} are selected.
The second set of delta modules (from line 27 to 34) corresponds to the integrations of the \textsf{Eval1} feature in the base program.
For instance, {\tt DreaddNegEval\_DNegEval1} is the $\named$ delta module created by the {\tt mergeToBase} function (line 35 in Figure~\ref{fig:refactoring:decreasing}),
 and does not contain any operations as the base program did not originally contain the {\tt eval} method;
 {\tt DremNegEval\_DNegEval1} is the $\named'$ delta module created by the {\tt mergeToBase} function (line 39 in Figure~\ref{fig:refactoring:decreasing}),
 and removes the $\name{Neg}.\name{eval}$ method when the feature \textsf{Eval1} or \textsf{Neg} is not selected.
\begin{figure}
\begin{minipage}[t]{.39\textwidth}
\begin{lstlisting}[numbers=left]
class Exp extends Object {
  String toString() { return ""; }
  int eval() { ... }
}
class Lit extends Exp {
  int value;
  Lit setLit(int n) { ... }
  String toString() { ... }
  int eval() { ... }
}
class Add extends Exp {
  Exp expr1;
  Exp expr2;
  Add setAdd(Exp a, Exp b) { ... }
  String toString() { ... }
  int eval() { ... }
}
class Neg extends Exp {
  Exp expr;
  Neg setNeg(Exp a) { ... }
  String toString() { ... }
  int eval() { ... }
}
\end{lstlisting}
\end{minipage}\begin{minipage}[t]{.59\textwidth}
\begin{lstlisting}[numbers=left,firstnumber=last]
DremNeg_DNeg { removes Neg }
DremNegToString_DNegPrint { modifies class Neg { removes toString } }

DreaddNegEval_DNegEval1 {  }
DremNegEval_DNegEval1 { modifies class Neg { removes eval } }
DreaddExpEval_DLitEval1 {  }
DremExpEval_DLitEval1 { modifies class Exp { removes eval } }
DreaddLitEval_DLitEval1 {  }
DremLitEval_DLitEval1 { modifies class Lit { removes eval } }
DreaddAddEval_DAddEval1 {  }
DremAddEval_DAddEval1 { modifies class Add { removes eval } }

DreaddNegEval_DNegEval2 {
  modifies class Neg { readds Lit eval() { ... } } }
DremNegEval_DNegEval2 {  }
DreaddExpEval_DLitEval2 {
  modifies class Exp { readds Lit eval() { ... } } }
DremExpEval_DLitEval2 {  }
DreaddLitEval_DLitEval2 {
  modifies class Lit { readds Lit eval() { ... } } }
DremLitEval_DLitEval2 {  }
DreaddAddEval_DAddEval2 {
  modifies class Add { readds Lit eval() { ... } } }
DremAddEval_DAddEval2 {  }
\end{lstlisting}
\end{minipage}\vspace*{-1em}
\caption{\EPL\ Modified by the Decreasing Refactor Algorithm}\label{fig:epl:decreasing}
\end{figure}
The last set of delta modules (from line 36 to 47) corresponds to the integrations of the \textsf{Eval2} feature in the base program.
As when including this feature in the base program, the delta modules for \textsf{Eval1} already have been integrated,
  the {\em readding} delta modules contains the implementation of the \textsf{Eval2} version of the {\tt eval} method;
  and on the opposite, the {\em removing} delta modules are empty.

Note that, since the original SPL contains method $\kwmodifies$ operations that are not $\kwvoids$,
 the refactored SPL is readds-pseudo-decreasing monotonic.% (cf.\ the discussion at the beginning of Section~\ref{sec:algo-dec-algo}).
On the other hand, since in the projection of the original \EPL\ on the products without feature \textsf{Eval2} each class/attribute is added only once (see Example~\ref{exa:prj-Eval} in Section~\ref{sec:model}),
 its decreasing monotonic refactoring would produce a pseudo-decreasing product line.

\vspace*{-.7em}
\subsection{Properties}\label{sec:algo-prop}

We finally present the main properties of these two refactoring algorithms.
As they both share the same characteristics, we state our theorems for both of them.
%A more indepth discussion on the monotonic properties of the generated \SPL\ is presented in the next section.
\begin{theorem}[Correctness]% and Completeness]
%For all \SPL\ $\sePIFDJ$, {\tt refactor($\sePIFDJ$)} is monotonic and has the same products and variants as $\sePIFDJ$.
Applying one of the {\tt refactor} algorithms on one \SPL\ $\sePIFDJ$ is a monotonic \SPL\ that have the same products and variants as $\sePIFDJ$.
\end{theorem}
\begin{proof}[Proof (sketch).]
Let us consider the increasing version of the {\tt refactor} algorithm (proving the result for the decreasing version is similar),
 and let us denote $\sePIFDJ'$ as $\texttt{refactor}(\sePIFDJ)$.
The fact that $\sePIFDJ'$ is monotonic is a direct consequence of the algorithm iterating over all delta operations and deleting all the $\kwremoves$ operations.
The fact that $\sePIFDJ'$ has the same products as $\sePIFDJ$ is a direct consequence of {\tt refactor} not changing the \FM\ of $\sePIFDJ$.
The fact that $\sePIFDJ'$ has the same variants as $\sePIFDJ$ can be proven by checking that each product $p$ of $\sePIFDJ$
 generates the same variant in $\sePIFDJ'$ and in $\sePIFDJ$:
 this can be done by induction on the number of delta modules and delta operations used to generate the variant of $p$ in $\sePIFDJ$.
%
%Proving that all the products of $\sePIFDJ$ can still be generated can be done by considering one product,
% checking that every delta module can be applied, depending if it was generated or not.
%Finally checking that a product corresponds to the same variant as before can be done by generating step by step the variant,
% and remarking that if at one step of the generation the partial variant differs from the original one, then there is a delta module later on that will remove this difference.
\end{proof}
\noindent
Recall that the notion of increasing (resp.\ decreasing) monotonicity satisfied by the refactored SPL depends on the properties of the original SPL,
 as pointed out at the beginning of Section~\ref{sec:algo-inc-algo} (resp.\ Section~\ref{sec:algo-dec-algo}).

\begin{theorem}[Complexity]
The space complexity of the {\tt refactor} algorithms is:
 i) constant in the size of \IFJ\ code;
 ii) linear in the number of delta operations;
 and iii) linear in the number of delta operations times the number of delta modules for the generation of the activation condition of the new delta modules.
The time complexity of the {\tt refactor} algorithms is quadratic in the number of delta operations.
\end{theorem}
\begin{proof}[Proof (sketch).]
i) is a direct consequence of the algorithm not creating or duplicating \IFJ\ code.
ii) is more subtle:
 in the increasing refactor, $\ado_1$ is replaced by one delta module containing one operation, and $\ado_2$ is kept as it is;
 however in the decreasing refactor, to match all the $\kwreadds$ and $\kwremoves$ operations generated in {\tt mergeToBase},
 we need to consider that adding a class corresponds to one $\kwadds$ operation for  the class name, and one $\kwadds$ operation for each of its fields.
iii) it is straightforward to see that the length of the activation condition of the delta module created in function {\tt mergeOperations}
 is linear in the number of delta modules in $\sePIFDJ$.
Finally, {\tt refactor} is quadratic in time in the number of delta operations as it iterates over them with two inner loops
 (one in function {\tt refactor}, one in function {\tt manageOperation}).
\end{proof}

\section{Related Work}\label{sec:related}
To the best of our knowledge, refactoring in the context of \DOP\ has been studied only in~\cite{Schulze:2013:RDS:2451436.2451446} and~\cite{DBLP:journals/corr/HaberRRS14}.
The former considers product lines of Java programs, while the latter considers delta modeling of software architectures.
We refer to~\cite{Schulze:2013:RDS:2451436.2451446} for the related work in the \FOP\ or annotative approaches.
Note that both of these approaches are monotonic by construction (\FOP\ being increasing, and annotative being decreasing),
  and so no refactoring algorithms to achieve monotonicity exist for them.
In~\cite{Schulze:2013:RDS:2451436.2451446}, a catalogue of refactoring and code smells is presented,
 and most of them focus on changing one delta module, one feature at a time.
Two of their refactorings are related to ours.
{\em Resolve Modification Action} replaces a $\kwmodifies$ operations that does not call $\kworiginal$ with an $\kwadds$ operation,
 by modifying the activation condition of previous $\kwmodifies$ and $\kwadds$ operations.
{\em Resolve Removal Action} eliminates $\kwremoves$ operations also by changing the application condition of previous $\kwmodifies$ and $\kwadds$ operations.
Other refactoring algorithms focus on how to enable extractive \SPL\ development for \FOP~\cite{Alves:2006,Liu:2006}.
These works are related to ours, as \DOP\ natively supports extractive \SPL\ development:
 refactoring such a \SPL\ into an increasing monotonic one using our algorithms is close to adapting this \SPL\ to \FOP.
%
%On the other hand, a theory of \SPL\ refinement have been studied in~\cite{Borba2010} that enables stepwise \SPL\ evolution.
%Finally, refactoring have also been studied in the context of {\em Aspect-Oriented Programming} in~\cite{Monteiro:2005} which presents some algorithms and code smells.
%However, this work does not consider the variability aspect of \SPL s.

\section{Conclusion and Future Work}\label{sec:conclusion}

In this paper, we presented two refactoring algorithms with the goal of lowering the complexity of the input \SPL,
 by removing opposite $\kwadds$ and $\kwremoves$ operations.
These algorithms work by removing one kind of operation from the input \SPL, either $\kwadds$ or $\kwremoves$,
 and so they do not duplicate code nor change the structure of the input \SPL, except for the parts related to the removed operation.

We plan four lines of future work for monotonicity in \DOP.
First, %in the {\em reactive} \SPL\ development approach, it can occur that parts of the base program cannot be modified, an so,
% some delta operations cannot be removed using the algorithms we presented in this article.
%It would be interesting
 we would like to investigate alternative means to reach (a possibly more flexible version of) monotonicity.
Second, complementarily to our algorithms, one could consider also refactoring code.
For instance, splitting the definition of a method into several ones would help into transforming $\kwmodifies$ operations in $\kwvoids$ operations.
Third, we would like to identify specific analysis scenarios where monotone product lines are simpler to analyze.
Fourth, we plan to develop case studies in order to evaluate the advantages and the drawbacks of the proposed refactorings.

\paragraph{Acknowledgements.} We are  grateful to the FMSPLE 2016 anonymous reviewers for many comments and suggestions for improving the presentation.

\vspace*{-.8em}

\bibliographystyle{abbrvurl}
\bibliography{biblio}

\ifappendix
\section{Appendix}\label{sec:appendix}

\begin{center}
\emph{{This appendix is not part of the submission\\ and it is included for referees' convenience only.}}
\end{center}

\bigskip

%\subsection{Temporary Definition (Before we know where to put them)}\input{src/99-temporary-definitions}

%\subsection{Proofs}\label{sec:proofs}
\input{src/80-proofs}

%\subsection{\IFJ\ Type System Rules}\label{sec:ifj-ts}\input{src/81-IFJ-ts}
%\subsection{\IFDJ\ Partial Type System Rules}\label{sec:ifdj-pts}\input{src/82-IFDJ-partial-ts}
%\subsection{\IFDJ Operational Semantics}\label{sec:ifdj-semantics}\input{src/85-variant-generation}
%\input{src/91-EPL-example-in-IFDJ}
\fi

\end{document}